\pdfoutput=1   
\documentclass[10pt]{article}
\usepackage[preprint]{tmlr}

\usepackage{amsmath,amssymb,amsthm}
\usepackage{booktabs}
\usepackage{graphicx}
\usepackage{xcolor}
\usepackage{url}
\usepackage{hyperref}
\hypersetup{hidelinks}   

\newtheorem{lemma}{Lemma}
\newtheorem{theorem}{Theorem}
\newcommand{\Piodd}{\Pi_{\mathrm{odd}}}
\newcommand{\Pieven}{\Pi_{\mathrm{even}}}
\newcommand{\expct}[1]{\langle #1 \rangle}

\title{Parity Floors in Quantum Denoisers: A Closed-Form Benchmark for Fixed-Map Denoising Networks}

\author{\name Jaeuk Kim \email freak91uk@hnextits.com \\ \addr NextITS, Seoul, Republic of Korea}

\begin{document}
\maketitle

\begin{abstract}
Fixed quantum feature maps are increasingly inserted into diffusion denoisers, but
standard image benchmarks do not reveal \emph{which} structural constraint limits
them. We introduce \textbf{CoupledPhaseTexture}, a torus-diffusion benchmark with
analytic heat-kernel noising that separates \emph{parity}, \emph{within-sector
approximation}, and \emph{sample-complexity} limitations. For the depth-1
$\mathrm{RY}$+CNOT+Pauli-$Z$ family we prove a containment-free parity floor: all
reachable features are even functions of the encoded angles while the sine
components of the Bayes denoiser are odd, so the \emph{excess} risk splits exactly
into an inaccessible odd part and a within-sector residual,
$R_\sigma(Q)-R^\star_\sigma=\|\Piodd m_\sigma\|^2+\inf_{q\in Q}\|\Pieven m_\sigma-q\|^2$.
The first term is an irreducible, $\sigma$-resolved lower bound holding for
\emph{every} even feature class, with no containment, linearity, or closedness
assumption on $Q$. The obstruction is a property of the noise-conditioned denoising
target rather than static representability: the floor is re-derived at each noise
scale because the target's parity content changes with $\sigma$. The measured excess
is dominated by the parity proxy on two distinct priors. Higher-order $Z$ readouts
improve the even sector, but entanglement does not lower the floor and re-uploading
does not reliably close it. Classical controls confirm the deficit is parity rather
than quantumness: a cosine-only bank is floored similarly, while adding the sine
sector matches the reference. Among tested constructions---odd readouts, a
noise-coupled encoder---none matches the sine-carrying classical bank. These results
motivate nonclassical data access or feature classes without efficient classical
surrogates; they do not establish either as sufficient for quantum advantage.
\end{abstract}

\section{Introduction}
Empirical claims that a small quantum circuit helps a generative model are hard
to interpret: inserting \emph{any} bottleneck helps a baseline, and the quantum
arm is rarely matched to a classical control of equal capacity, dimension, or
spectral reach. When the circuit is classically simulable, the interesting
question is not asymptotic advantage but \emph{which structural constraint of the
map limits it on the task}---a question the prevailing image-benchmark practice
cannot answer, because a raw FID number does not separate ``the map cannot
represent the target'' from ``the map overfits'' from ``the baseline was weak.''
A prior fair-comparison study of \emph{variational} quantum cores in image
diffusion \citep{paper1} established parity with matched classical controls but,
being an image bake-off, left the mechanism unresolved.

We give a controlled, closed-form instrument. Our object of study is the
\emph{fixed} quantum feature map---an angle-encoded circuit with no trainable
gate parameters, read out as $\expct{Z_i}$ and $\expct{Z_iZ_j}$---and our testbed
is a synthetic diffusion process on the torus whose Bayes-optimal denoiser is
computable to reference accuracy. On it, for parity-invariant fixed maps the
measured excess admits a clean diagnostic attribution, and one component is a
theorem. By ``closed-form'' we mean the
forward heat-kernel noising, the parity decomposition, and the odd-sector lower
bound are analytic; the numerical risk tables use a finite-sample numerical
trigonometric reference (fit at large $N$), not an analytic oracle.

\paragraph{Contributions.}
\begin{enumerate}
\item \textbf{CoupledPhaseTexture} (Sec.~\ref{sec:bench}): a torus-diffusion
benchmark with analytic heat-kernel noising, a \emph{selected} finite-sample
numerical reference (degree-3, selected on its performance at representative noise
levels among degrees $2$--$5$ and retained for protocol consistency; not the
pointwise minimizer at every $\sigma$), and a matched, capacity-noted baseline suite.
\item \textbf{A containment-free parity error-floor theorem} (Sec.~\ref{sec:theory}):
single-layer angle+$Z$ maps are confined to the even sector; hence
$R_\sigma(Q)-R^\star_\sigma\ge\|\Piodd m_\sigma\|^2$ for every $\sigma$, with the
right-hand side an exact $\sigma$-resolved lower bound equal to the odd-sector $L^2$
mass of the heat-kernel-smoothed Bayes denoiser. It is
the $\mathrm{RY}$/diffusion specialization of Fourier-view dequantization
\citep{schuld2021fourier}; the new content is the odd-sector floor and its
$\sigma$-profile.
\item \textbf{A reference-relative diagnostic attribution instrument}
(Sec.~\ref{sec:instrument}): separate a map's measured excess into an empirical
parity proxy and a signed empirical residual, with finite-sample effects assessed
independently through an $N$-sweep. ``Which constraint binds'' is then empirical
and the instrument is reusable across parity-invariant fixed encoder choices that
expose comparable feature vectors.
\item \textbf{Finding: the depth-1 RY+CNOT+$Z$ family is parity-dominated}
(Sec.~\ref{sec:results}): the measured excess is dominated by the parity proxy
(small absolute signed residual, $|\hat\delta_{\sigma,N}|\le0.018$, Table~\ref{tab:ci}) on
\emph{two} structurally distinct priors (a 3-body coupled-phase texture and a
frustrated pairwise XY model), while the tested extensions---entanglement variants
and re-upload depth---do not close the gap. A short conditional-image check does not
contradict the prediction (Sec.~\ref{sec:image}).
\end{enumerate}
Our aim is characterization, not promotion: the deliverable is a reusable
instrument that pinpoints where a fixed map's expressivity ends and attributes the
residual deficit to a named cause.

\section{The CoupledPhaseTexture benchmark}
\label{sec:bench}
\paragraph{Data.} A configuration is $n{=}8$ phases $\theta\in\mathbb{T}^8$ drawn
from a coupled Gibbs law $p(\theta)\propto\exp\!\big(K\sum_{(i,j,k)\in\mathcal T}
\cos(\theta_i+\theta_j-2\theta_k)\big)$ on a random triple set $\mathcal T$
(a genuinely multi-body coupling, so local marginals are insufficient). The
prior is invariant under $\theta\mapsto-\theta$.

\paragraph{Forward process.} Wrapped-Gaussian angle noise
$\theta_\sigma=\theta_0+\sigma\varepsilon \bmod 2\pi$, $\varepsilon\sim\mathcal N(0,I)$
---the heat kernel on the torus, attenuating Fourier mode $k$ by
$e^{-\sigma^2k^2/2}$.

\paragraph{Task and reference.} Denoise: predict $Y=(\cos\theta_0,\sin\theta_0)
\in\mathbb{R}^{16}$ from $\theta_\sigma$ (the Bayes predictor is the denoiser
$m_\sigma=\mathbb E[Y\mid\theta_\sigma]$, not the noisy-density score). The exact
Bayes risk $R^\star_\sigma=\mathbb E\,\mathrm{Var}(Y\mid\theta_\sigma)$ is unknown;
we use a \emph{numerical reference} $R_{\mathrm{ref},\sigma}$, the risk of a
degree-3 both-parity trigonometric bank (CV-selected ridge, $N{=}8000$). We select
this bank on its performance at representative noise levels---degrees $4$--$5$ raise
the risk through estimation variance (Appendix~\ref{app:conv})---and retain it
across the evaluation grid for protocol consistency; it is \emph{not} necessarily
the pointwise risk minimizer at every noise level or prior (at large $\sigma$ the
degree-2 bank can be lower). It satisfies
$R_{\mathrm{ref},\sigma}\ge R^\star_\sigma$ (the exact gap to Bayes is unknown) and
serves as a common finite-sample reference for all compared maps, fit
\emph{independently} of any of them. All excess risks are reported relative to
$R_{\mathrm{ref},\sigma}$; readout is ridge (with held-out $\lambda$ selection),
removing optimization and regularization confounds.

\paragraph{Baseline suite (matched).} Fixed quantum map
($\expct{Z_i}$, $+\expct{Z_iZ_j}$, $+\expct{Z_iZ_jZ_k}$; entangler and depth
varied); a selected finite-sample even (cosine-only) reference; a
sin-carrying full bank; and classical even/full banks with feature counts noted
alongside the quantum maps so parity is not confounded with capacity.
Seeds and configuration files are specified in Appendix~\ref{app:repro}, and the code
will be released publicly.

\section{The parity error floor}
\label{sec:theory}
On $\mathbb{T}^n$ with the uniform measure, ``even/odd'' is parity under
$a\mapsto-a$ componentwise; the cosine (even) and sine (odd) sectors are
$L^2$-orthogonal. All excess-risk projections below are taken in $L^2$ over the
marginal law of the noisy phase $\theta_\sigma$; because both the prior and the
wrapped-Gaussian kernel are invariant under $\theta\mapsto-\theta$, this marginal
is itself parity-symmetric, so the even and odd subspaces remain orthogonal under
it (not only under the uniform measure).

\begin{lemma}[Even-sector confinement, single encoding layer]\label{lem:parity}
Let $q(a)=\big(\expct{Z_S}\big)_S$ with
$\expct{Z_S}(a)=\langle0|\,U^\dagger \textstyle\prod_{i\in S}Z_i\,U\,|0\rangle$,
$U=U_{\mathrm{ent}}\bigotimes_i \mathrm{RY}(a_i)$, and $U_{\mathrm{ent}}$ any
circuit of CNOTs. Then each $\expct{Z_S}(a)=\prod_{i\in S'}\cos a_i$ for some
$S'$, hence $q$ is even and $\mathrm{span}\,q\perp$ the sine sector.
\end{lemma}
\begin{proof}
$\mathrm{RY}(a_i)|0\rangle$ is real with $\expct{Z}=\cos a_i$, and CNOT is a real
permutation, so $U|0\rangle$ is real. Conjugating an all-$Z$ string by a CNOT
gives $Z_c\!\mapsto\!Z_c,\ Z_t\!\mapsto\!Z_cZ_t$: still all-$Z$, never producing
$X$ or $Y$. Thus $U^\dagger Z_S U=Z_{S'}$ and
$\expct{Z_S}(a)=\prod_{i\in S'}\cos a_i$, which is even; evenness survives linear
combination.
\end{proof}

\begin{theorem}[Containment-free parity floor]\label{thm:floor}
Let $Q$ be any even-sector feature class and $m_\sigma(\cdot)=\mathbb E[Y\mid
\theta_\sigma=\cdot]$ the Bayes denoiser. Since the prior is $\theta\mapsto-\theta$
symmetric and the noise is symmetric, the $\cos\theta_0$-components of $m_\sigma$
are even and the $\sin\theta_0$-components are odd. Then the excess risk splits \emph{exactly} into a parity term and a
within-sector approximation term,
\[
R_\sigma(Q)-R^\star_\sigma
=\underbrace{\big\|\Piodd m_\sigma\big\|^2_{L^2(\mu_\sigma)}}_{\text{parity floor}}
+\underbrace{\inf_{q\in Q}\big\|\Pieven m_\sigma-q\big\|^2_{L^2(\mu_\sigma)}}_{\text{within-sector residual}},
\]
and in particular
\[
R_\sigma(Q)-R^\star_\sigma \;\ge\; \big\|\Piodd m_\sigma\big\|^2
\;=\;\sum_{i}\mathbb E_{\theta_\sigma}\!\big[(\,\mathbb E[\sin\theta_0^i\mid\theta_\sigma]\,)^2\big],
\]
with equality exactly when $\inf_{q\in Q}\|\Pieven m_\sigma-q\|^2=0$, i.e.\ when the even
component of the Bayes target is approximable arbitrarily well within $Q$ (for closed $Q$,
when $\Pieven m_\sigma\in Q$). The bound needs \emph{no}
containment of $Q$ in any particular bank, and no linearity or closedness of $Q$;
only $Q\subseteq\{\text{even}\}$.
\end{theorem}
\begin{proof}
Write $m_\sigma=\Pieven m_\sigma+\Piodd m_\sigma$. For any $q\in Q\subseteq\{\text{even}\}$
the odd component is $L^2(\mu_\sigma)$-orthogonal to $m_\sigma-q$'s even part, so
$\|m_\sigma-q\|^2=\|\Piodd m_\sigma\|^2+\|\Pieven m_\sigma-q\|^2$. Taking the infimum over
$q\in Q$ leaves the first term untouched and gives the stated decomposition; the inequality
follows since the second term is non-negative. Note that $Q$ need not equal the whole even
sector --- the parity term is the same for every even $Q$, and only the second term depends
on which even features are reachable. Nor is $Q$ assumed closed or linear: the
identity holds pointwise in $q$ and the infimum is taken afterwards.
\end{proof}

\paragraph{Scope (honest).} (1) Lemma~\ref{lem:parity} is for a \emph{single}
encoding layer; data re-uploading enlarges the reachable set to both-parity
trigonometric polynomials \citep{schuld2021fourier}, and we verify numerically
that depth${\ge}2$ breaks the exact evenness ($\max_a\|q(a)-q(-a)\|$ jumps from
$0$ to $\mathcal O(1)$). Theorem~\ref{thm:floor} then applies to the depth-1 map;
Sec.~\ref{sec:results} shows depth${\ge}2$ nonetheless fails to close the floor.
(2) The floor is \emph{not} a $\sigma$-coupled separation: $U$ is $\sigma$-invariant,
so all $\sigma$-dependence lives in the target $m_\sigma$; the $\sigma$-profile of
$\|\Piodd m_\sigma\|^2$ is a property of the diffused denoiser, cited to the classical
noise-scale literature, not of the quantum map. (3) We do \emph{not} claim
$Q\subseteq$ a degree-2 bank---that is false (CNOT readouts reach higher-degree
even monomials); the theorem uses only even-sector membership. In short, the floor
theorem is a statement about the depth-1 $\mathrm{RY}$+CNOT+$Z$ family; re-uploaded
encoders and non-$Z$ readouts are treated as separate empirical variants, not
covered by it.

\section{The reference-relative diagnostic instrument}
\label{sec:instrument}
Given a fixed map, split its measured excess (via large-$N$ ridge, $\approx$
population risk) into
\[
\underbrace{R_\sigma^N(Q)-R_{\mathrm{ref},\sigma}}_{\text{measured excess}}
=\underbrace{\big(R_\sigma^{\text{even-ref}}-R_{\mathrm{ref},\sigma}\big)}_{\text{parity proxy }\hat P_\sigma}
+\underbrace{\big(R_\sigma^{N}(Q)-R_\sigma^{\text{even-ref}}\big)}_{\text{signed empirical residual }\hat\delta_{\sigma,N}},
\]
where $R_{\mathrm{ref}}$ is the selected numerical reference and
$R_\sigma^{\text{even-ref}}$ the selected finite-sample even reference, both fit at
$N{=}8000$ and independent of $Q$. The empirical \emph{parity proxy}
$\hat P_\sigma=R_\sigma^{\text{even-ref}}-R_{\mathrm{ref},\sigma}$ estimates the
theoretical parity floor $P_\sigma=\|\Piodd m_\sigma\|^2$ of Theorem~\ref{thm:floor}
up to finite-reference approximation and estimation errors. The second term is a \emph{signed empirical residual}
$\hat\delta_{\sigma,N}=R_\sigma^N(Q)-R_\sigma^{\text{even-ref}}$: it combines the
population reference-relative residual with the map's finite-sample effect, and may
be signed because the finite even reference need not contain the feature span of $Q$.
Its sign combines within-parity mismatch and finite-sample estimation effects and
does not by itself identify which dominates; persistence under the $N$-sweep
(Table~\ref{tab:e5}) supports a representational reading. It is the residual of
Tables~\ref{tab:head}--\ref{tab:ci}, negative at $\sigma{=}1.4$. Reading: measured
excess $\approx\hat P_\sigma\Rightarrow$ parity-dominated; a resolved
$\hat\delta_{\sigma,N}$ indicates a measurable difference between the two finite
feature classes, whose $N$-dependence distinguishes finite-sample effects from a
persistent representational mismatch.
``Which constraint binds'' is measured, and the instrument applies to
parity-invariant fixed encoder/readout/entangler families that expose comparable
feature vectors.

\section{Benchmark results}
\label{sec:results}
Averages over 8 random 3-body graphs (decomposition) and 6 graphs (design axes);
ridge with held-out $\lambda$ selection; $N{=}8000$ (decomposition) unless swept.

\paragraph{The depth-1 RY+CNOT+$Z$ map is parity-dominated (headline).}
Table~\ref{tab:head} and Fig.~\ref{fig:floor}: measured against the numerical
reference $R_{\mathrm{ref}}$, the depth-1 quantum excess is dominated by the parity
proxy $\hat P_\sigma$ at every $\sigma$. The signed empirical residual
$\hat\delta_{\sigma,N}$ is small in absolute value ($|\hat\delta_{\sigma,N}|\le0.018$;
per-graph paired $95\%$ bootstrap CIs in Table~\ref{tab:ci}) though statistically
resolved (its CI excludes zero); it is even slightly \emph{negative} at
$\sigma{=}1.4$, where the finite even reference does not contain the quantum even span. So the map is parity-\emph{dominated}: its excess is
almost, but not exactly, the odd sector it cannot represent. The parity proxy decreases with $\sigma$ as the heat kernel washes out
odd structure. That the dominant deficit is parity---not the map being quantum, nor
capacity---is shown directly in Section~\ref{sec:positive}: a purely \emph{classical}
even (cosine-only) bank is floored to the same level, while adding the sine sector
reaches the reference.

\begin{table}[t]\centering
\caption{Excess risk vs the numerical reference $R_{\mathrm{ref}}$ (degree-3 trig
bank, CV-selected ridge, $N{=}8000$; 8 graphs; references fit independently of the
compared map). Depth-1 quantum excess is dominated by the parity proxy
$\hat P_\sigma$ (selected even reference $-$ $R_{\mathrm{ref}}$) at every $\sigma$;
the signed empirical residual $\hat\delta_{\sigma,N}$ is small and is quantified with
paired bootstrap CIs in Table~\ref{tab:ci}.}
\label{tab:head}
\begin{tabular}{lccc}
\toprule
$\sigma$ & parity proxy $\hat P_\sigma$ & quantum excess & empirical residual $\hat\delta_{\sigma,N}$ \\
\midrule
0.20 & 0.250 & 0.263 & 0.013 \\
0.35 & 0.234 & 0.250 & 0.016 \\
0.50 & 0.212 & 0.230 & 0.018 \\
0.70 & 0.171 & 0.190 & 0.018 \\
1.00 & 0.101 & 0.109 & 0.008 \\
1.40 & 0.030 & 0.024 & $-0.006$ \\
\bottomrule
\end{tabular}
\end{table}

\begin{table}[t]\centering
\caption{Headline uncertainty (8 graphs). The signed empirical residual
$\hat\delta_{\sigma,N}=$ quantum excess $-$ parity proxy is reported relative to the
selected finite-sample even reference and may be \emph{signed} (it is negative at
$\sigma{=}1.4$), because that reference need not contain the quantum even feature
span. We give its mean, standard error, and a $5000$-resample bootstrap $95\%$ CI.
It is small in absolute value ($|\hat\delta_{\sigma,N}|\le0.018$) though its CI
excludes zero---so the map is parity-\emph{dominated}, not exactly parity-bound.}
\label{tab:ci}
\begin{tabular}{lccccc}
\toprule
$\sigma$ & parity proxy $\hat P_\sigma$ & quantum excess & residual $\hat\delta_{\sigma,N}$ & SE & bootstrap $95\%$ CI \\
\midrule
0.20 & 0.250 & 0.263 & 0.0133 & 0.0013 & $[+0.011,+0.016]$ \\
0.35 & 0.234 & 0.250 & 0.0160 & 0.0011 & $[+0.014,+0.018]$ \\
0.50 & 0.212 & 0.230 & 0.0183 & 0.0011 & $[+0.016,+0.020]$ \\
0.70 & 0.171 & 0.190 & 0.0181 & 0.0010 & $[+0.016,+0.020]$ \\
1.00 & 0.101 & 0.109 & 0.0079 & 0.0004 & $[+0.007,+0.009]$ \\
1.40 & 0.030 & 0.024 & $-0.0059$ & 0.0004 & $[-0.007,-0.005]$ \\
\bottomrule
\end{tabular}
\end{table}

\begin{figure}[t]\centering
\includegraphics[width=0.9\linewidth]{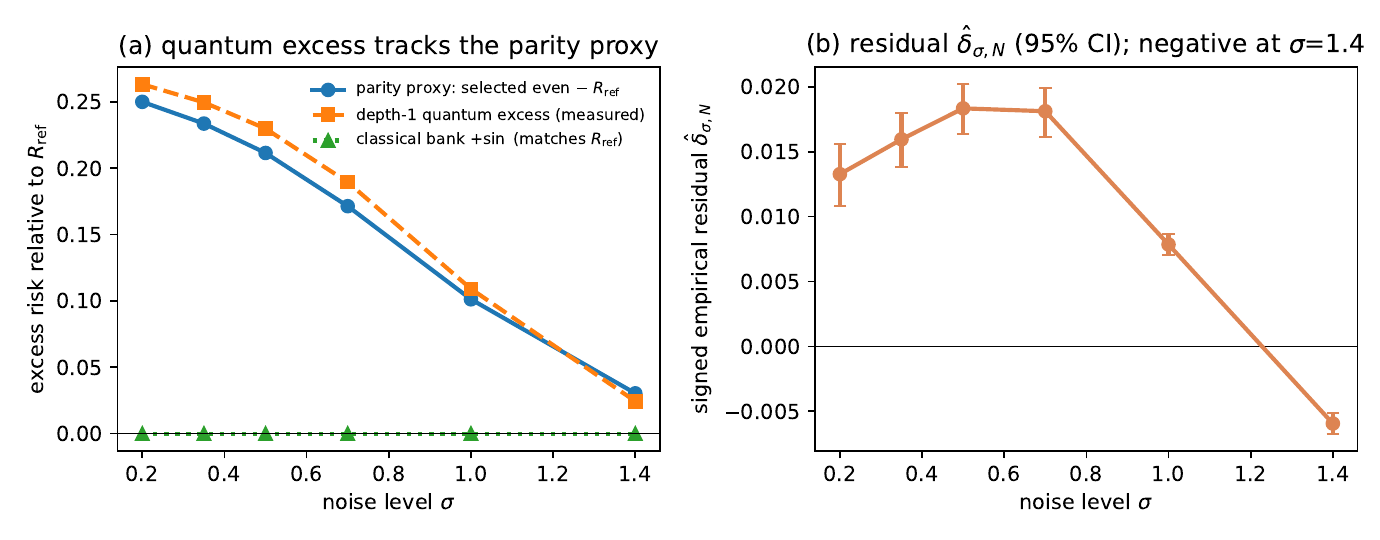}
\caption{(a) The measured depth-1 quantum excess closely tracks the parity proxy
(selected even reference $-$ $R_{\mathrm{ref}}$) $\approx\|\Piodd m_\sigma\|^2$ across
noise levels, while a classical bank that adds the sine sector matches
$R_{\mathrm{ref}}$. (b) The signed empirical residual $\hat\delta_{\sigma,N}$ is small
in absolute value ($\le0.018$); error bars are paired bootstrap $95\%$ CIs
(Table~\ref{tab:ci}). The negative value at $\sigma{=}1.4$ reflects that the finite
even reference need not contain the quantum even feature span.}
\label{fig:floor}
\end{figure}

\paragraph{The observed parity-proxy profile is consistent with heat-kernel attenuation.} Theorem~\ref{thm:floor}
makes the floor a function of the heat-kernel-attenuated odd spectrum: the forward
process damps Fourier mode $k$ by $e^{-\sigma^2k^2/2}$, so a floor dominated by the
leading odd harmonic should decay as $\|\Piodd m_\sigma\|^2\propto e^{-a\sigma^2}$
with $a=\mathcal O(1)$. Fitting the measured parity proxy to a two-parameter
$c\,e^{-a\sigma^2}$ over the tested $\sigma$ range gives $R^2{=}0.99$ on
CoupledPhaseTexture ($a{=}1.10$) and $R^2{=}0.997$ on the frustrated-XY prior
($a{=}0.65$; Table~\ref{tab:xy}); the fitted exponents are of the same order as the
mode-1 value $a{=}1$.
The fitted profile is thus \emph{consistent with} a leading odd harmonic dominating
over the tested range---a descriptive fit, not an exact closed-form identity for the
conditional denoiser---and this $\sigma$-dependence is diffusion-specific, with no
analogue in the static Fourier/expressivity view of a fixed encoding.

\paragraph{Sample axis.} Table~\ref{tab:e5}: apart from the very small-$N$ regime
(where finite-sample conditioning separates them), the depth-1 quantum map and its
even-sector classical equivalent track closely by $N\ge300$ and approach similar
reference-relative excess, both remaining separated from $R_{\mathrm{ref}}$ by a
parity-dominated component while the sin-carrying bank sits near $R_{\mathrm{ref}}$---so
the deficit is representational
(parity), not a
sample-complexity artifact.

\begin{table}[t]\centering
\caption{Sample-complexity ($\sigma{=}0.5$, 6 graphs; excess vs $R_{\mathrm{ref}}$). Quantum
and even-classical share sample complexity; both remain at similar
parity-dominated excess above the sine-carrying bank.}
\label{tab:e5}
\begin{tabular}{lccc}
\toprule
$N$ & quantum (depth-1) & even-classical & classical-full (sin) \\
\midrule
100  & 0.304 & 0.323 & 0.225 \\
300  & 0.264 & 0.260 & 0.085 \\
1000 & 0.243 & 0.229 & 0.031 \\
3000 & 0.235 & 0.218 & 0.020 \\
\bottomrule
\end{tabular}
\end{table}

\paragraph{Design axes.} Table~\ref{tab:e6}: higher-body readout captures more of
the even sector (order-1 $\to$ 2 $\to$ 3: $0.40\to0.23\to0.22$); the entangler
barely changes the reference-relative excess (\texttt{none}/\texttt{line}/\texttt{ring}
$\approx0.21$--$0.23$, all-to-all \emph{worse}, $0.34$); and re-upload
depth---which breaks the even confinement and so \emph{could} represent the odd
sector---does not close the floor (depth $1/2/3\!:\ 0.23/0.23/0.40$, non-monotone).
The non-monotone depth trend reflects representation misalignment, not capacity:
re-uploading pushes the fixed $36$-dim readout to high frequency, so train and test
excess rise together while conditioning \emph{improves}
(Appendix~\ref{app:design}); we therefore do not read re-uploading as a reliable
capacity gain, and entanglement is not what binds the map.

\begin{table}[t]\centering
\caption{Design-axis ablation ($\sigma{=}0.5$, $N{=}3000$, 6 graphs; excess vs
$R_{\mathrm{ref}}$, mean$\pm$std where shown; the entangler row reports means over 6 graphs).}
\label{tab:e6}
\begin{tabular}{lc@{\hskip 2em}lc}
\toprule
readout / entangler & excess & re-upload depth & excess \\
\midrule
order-1 $\expct{Z_i}$        & $0.397\pm0.005$ & depth 1 & $0.234\pm0.005$ \\
order-2 $+\expct{Z_iZ_j}$    & $0.234\pm0.005$ & depth 2 & $0.227\pm0.012$ \\
order-3 $+\expct{Z_iZ_jZ_k}$ & $0.219\pm0.003$ & depth 3 & $0.400\pm0.004$ \\
entangler none / line / ring / all-to-all  & \multicolumn{3}{l}{$0.215$ / $0.214$ / $0.234$ / $0.342$} \\
\bottomrule
\end{tabular}
\end{table}

\paragraph{Generalization to a second, structurally distinct prior.}
To check that the floor is a property of the diffused denoiser rather than of the
specific 3-body coupling, we re-run the \emph{entire} decomposition on a
\emph{2-body frustrated signed-XY} prior,
$p(\theta)\propto\exp\!\big(\beta\sum_{(i,j)\in E}J_{ij}\cos(\theta_i-\theta_j)\big)$
on a random 3-regular graph with $J_{ij}\in\{\pm1\}$ ($\beta{=}2.0$)---pairwise
rather than 3-body, frustrated rather than ferromagnetic, but still
$\theta\mapsto-\theta$ symmetric, so Theorem~\ref{thm:floor} applies. Table~\ref{tab:xy}
reproduces the CoupledPhaseTexture picture: the depth-1 quantum excess tracks the
parity proxy (mean ratio $1.00$, signed empirical residual $\approx0$), the classical
even bank is floored at the same level ($+0.175$ averaged over $\sigma$), and the
classical full bank matches the numerical reference ($-0.004$). Degrees $4$--$5$ do
not lower the reference here either (the 2-body prior carries little high-degree
structure), so degree-3 is again the selected reference. The proxy-tracking pattern
is thus not specific to one coupling, and the parity-floor theorem applies to both
symmetric priors.

\begin{table}[t]\centering
\caption{Second benchmark (2-body frustrated signed-XY, $\beta{=}2.0$, $N{=}8000$,
6 graphs; excess vs the numerical reference $R_{\mathrm{ref}}$). The decomposition reproduces
the CoupledPhaseTexture picture on a structurally different prior: quantum excess
$\approx$ parity proxy, classical-even floored, classical-full $\approx R_{\mathrm{ref}}$
(signed residuals within $\pm0.01$ read as zero, as in Table~\ref{tab:head}).}
\label{tab:xy}
\begin{tabular}{lccccc}
\toprule
$\sigma$ & parity proxy & quantum excess & empirical residual & classical-even & classical-full \\
\midrule
0.20 & 0.246 & 0.249 & 0.003 & 0.241 & $-0.001$ \\
0.35 & 0.230 & 0.233 & 0.002 & 0.225 & $-0.001$ \\
0.50 & 0.213 & 0.215 & 0.002 & 0.208 & $-0.001$ \\
0.70 & 0.186 & 0.187 & 0.001 & 0.182 & $-0.001$ \\
1.00 & 0.139 & 0.140 & 0.001 & 0.134 & $-0.004$ \\
1.40 & 0.069 & 0.067 & $-0.003$ & 0.061 & $-0.013$ \\
\bottomrule
\end{tabular}
\end{table}

\section{Image external-validity transfer}
\label{sec:image}
As a lightweight external-validity sanity check---not a full image-generation
benchmark---we inject the same fixed
$0$-trainable-parameter map ($\expct{Z_i}{+}\expct{Z_iZ_j}$, $36$ features;
distinct from the $\expct{Z_i}$-only variational cores of \citealp{paper1}) into
the conditioning path of a class-conditional diffusion model, and report
classifier \emph{controllability} (class-fidelity)---never unconditional FID.
On a $3$-class set of structural (civil) images generated
\emph{class-conditionally} (Table~\ref{tab:image}),
controllability saturates at $\approx1.0$ for the quantum fixed map, a
parameter-matched classical fixed map, and the no-map baseline alike, and the
best class-fidelity/FID is classical: no quantum edge, as the parity analysis
predicts. The map is inserted at the same bottleneck-conditioning hook used in a
prior fair-comparison study \citep{paper1}, purely as a transfer check; that
study's scaffold, tables, and parameter accounting are not re-reported here.

\begin{table}[t]\centering
\caption{Image external validity (conditional DDPM, $3$-class class-conditional
generation of structural civil images, classifier acc $1.00$, $500$/class). The
fixed quantum map shows no controllability or FID edge over a parameter-matched
classical fixed map; both are near-saturated. Headline metric is controllability,
not FID.}
\label{tab:image}
\begin{tabular}{lcc}
\toprule
conditioning map & controllability $\uparrow$ & FID $\downarrow$ \\
\midrule
none                       & 0.999 & 14.20 \\
classical SE core (trainable) & 1.000 & 12.91 \\
quantum fixed map $\expct{Z_i}{+}\expct{Z_iZ_j}$ & 1.000 & 13.51 \\
classical fixed map (param-matched) & 0.999 & \textbf{12.65} \\
\bottomrule
\end{tabular}
\end{table}

\section{Positioning: dequantization and the simulability trap}
\label{sec:boundary}
Even-sector confinement is a diffusion-specific instance of dequantization: the
single-layer angle map realizes a truncated-Fourier/kernel class
\citep{schuld2021fourier}, and the depth/non-Clifford taxonomy of
\citep{taxonomy2025} places it in the classically-tractable tier. The delta over
that line of work is explicit. Fourier/expressivity results characterize \emph{which
functions} a fixed encoding can represent and \emph{whether} it is classically
simulable---a model-side, task-agnostic statement. We instead give, for the
diffusion denoising task, (i) a \emph{task-resolved} lower bound on excess risk
that is exactly the odd-sector mass of the diffused Bayes denoiser, (ii) an
\emph{empirical $\sigma$-profile} consistent with leading-mode heat-kernel
attenuation ($c\,e^{-a\sigma^2}$, $R^2\!\ge\!0.99$ on two priors; descriptive, not
an exact identity), and (iii) a reference-relative attribution into a parity proxy
and a signed empirical residual, with finite-sample behavior assessed separately,
that turns a null into a diagnosis. None of these follow
from the static representability statement: an encoding can be ``expressive enough''
in the Fourier sense yet still pay this parity floor because the diffused target's
odd sector is orthogonal to the whole even class. A prominent regime not addressed by classical feature-map dequantization is
learning from quantum data \citep{huang2022experiments}. Separations in that setting
concern nonclassical data-access costs that are absent from classically tractable
simulated instances \citep{tensofqubits2026}: in any instance we can \emph{simulate},
privileged access to the known state-preparation circuit can eliminate the intended
measurement-access separation. This motivates
distinguishing simulability from genuinely nonclassical data access, and treating
simulated quantum-data experiments cautiously; we do not establish that either
nonclassical data access or an unmatched feature class is by itself sufficient for
advantage.

\paragraph{Related work.} Dequantization of kernel/feature-map QML began with
sampling-based low-rank methods and continues through random-Fourier-feature
surrogation of angle-encoded models \citep{schuld2021fourier,rffdequant2025}; our
floor is the diffusion-denoiser specialization. Fair-comparison and benchmarking
critiques argue that reported QML advantages often vanish against capacity-matched
classical baselines \citep{bowles2024benchmarking,schuld2022advantage,matchedspectral2026};
we contribute a closed-form testbed and a decomposition that turn such negatives
into calibrated measurements. Quantum and hybrid generative/diffusion models
\citep{zhang2024quddpm,parigi2024qndgdm,qgen2026} motivate the fixed-map object we
analyze. Learning from quantum data \citep{huang2022experiments,tensofqubits2026}
is the regime our boundary argument leaves open. Finally, classical diffusion and
denoising theory \citep{ho2020ddpm,song2021sde} supplies the objects we build
on---the conditional denoiser, heat-kernel smoothing on the (here toroidal) domain
\citep{debortoli2022riemannian}, and the resulting spectral attenuation; our
contribution is not a new diffusion identity but a task-resolved parity lower bound
and a diagnostic benchmark for fixed quantum feature maps in this setting.

\section{What a positive would likely require in this setting}
\label{sec:positive}
Closing the floor requires odd content where the denoiser needs it. Measuring every
map against the \emph{independent} numerical reference $R_{\mathrm{ref}}$ and noting feature
counts (Table~\ref{tab:sigma}) separates parity from both quantumness and capacity.
Two controls are decisive. First, a purely \emph{classical} even (cosine-only)
bank is floored to essentially the same level as the quantum even map
($+0.165$ vs.\ $+0.178$ averaged over $\sigma$)---so the floor is a \emph{parity}
property, not a quantum defect. Second, the classical \emph{full} bank (the same
features plus the sine sector) matches the numerical reference ($+0.004$), confirming
that the missing ingredient is odd content, not capacity. Against this, giving the
quantum map odd ($\expct{X_i}$) readouts (a $100$-feature map) retains a mean excess
of $+0.144$ over the six-noise grid, and a $\sigma$-coupled variant
($\mathrm{RZ}(\phi(\sigma){=}\sigma)$ before $\expct{Y}$ readouts, $72$ features)
changes it only marginally ($+0.176$ mean, a $\approx0.003$ reduction of the
$q_{\text{even}}$ excess; Table~\ref{tab:sigma}): this is consistent with the
resulting entangled readout basis being poorly aligned with the denoiser's relevant
low-frequency statistics. Among the in-simulation constructions tested here, neither odd-augmented readouts
nor noise-schedule-coupled variants match the sine-carrying classical bank (which
matches $R_{\mathrm{ref}}$). This motivates future study of beyond-simulation quantum
data \citep{huang2022experiments,tensofqubits2026} and feature classes without
efficient matched classical surrogates.

\begin{table}[t]\centering
\caption{Excess vs the \emph{independent} numerical reference $R_{\mathrm{ref}}$ (6 graphs);
feature counts in the header separate parity from capacity. A classical even bank
($64$f, no sin) is floored like the quantum even map ($36$f); the classical full
bank ($128$f, $+\sin$) matches $R_{\mathrm{ref}}$; the odd-augmented quantum map ($100$f)
partially reduces the excess, whereas the $\sigma$-coupled RZ variant ($72$f) changes
it only marginally.}
\label{tab:sigma}
\begin{tabular}{lccccc}
\toprule
$\sigma$ & $q_{\text{even}}$ (36f) & $q_{\text{full}}$ (100f, $X{+}Z$) & $q_{\sigma}$ (72f) & $c_{\text{even}}$ (64f) & $c_{\text{full}}$ (128f, $+\sin$) \\
\midrule
0.20 & 0.263 & 0.204 & 0.257 & 0.244 & \textbf{0.001} \\
0.35 & 0.250 & 0.197 & 0.244 & 0.232 & \textbf{0.008} \\
0.50 & 0.230 & 0.186 & 0.227 & 0.214 & \textbf{0.016} \\
0.70 & 0.189 & 0.158 & 0.187 & 0.176 & \textbf{0.017} \\
1.00 & 0.110 & 0.095 & 0.110 & 0.102 & \textbf{0.000} \\
1.40 & 0.028 & 0.023 & 0.028 & 0.024 & \textbf{$-$0.018} \\
\bottomrule
\end{tabular}
\end{table}

\section{Conclusion}
Within simulable scale we contribute a closed-form benchmark and a
decomposition instrument that turn ``no quantum advantage'' from an unexplained
null into a calibrated, mechanistic reading: for the depth-1 RY+CNOT+$Z$ family the
dominant limitation is parity confinement, with a small reference-relative residual
under the tested finite-sample protocol; the tested extensions---odd-augmented
readouts, entanglement variants, re-upload depth, and $\sigma$-coupled
encoders---do not close the gap to a matched classical sine-carrying bank. The
benchmark can evaluate a broad class of fixed feature maps, while the parity
attribution developed here applies to parity-invariant feature families; together
they delimit the classically simulable fixed-map regime studied here and motivate,
rather than establish, future investigation of nonclassical data access and feature
classes not matched by a simple classical spectral surrogate.

\paragraph{Author Contribution.}
Jaeuk Kim conceived the study, developed the theoretical analysis, designed
and conducted the experiments, analyzed the results, and prepared the
manuscript. Generative AI tools were used to assist with reviewing manuscript
revisions, with rewording and condensing passages of the manuscript from the
author's revision notes, and with \LaTeX{} typesetting.
All scientific content, mathematical arguments, experimental
design, results, and conclusions originate with the author, who independently
reviewed and verified every statement in the manuscript.

\section*{Acknowledgements}
This research was supported by Seoul R\&BD Program (QR250005, ``Building a Quantum
Computing--AI Integrated Development Platform and Developing Quantum--AI Algorithm'')
through the Seoul Business Agency (SBA) funded by Seoul Metropolitan Government.

\bibliographystyle{tmlr}
\bibliography{main}

\appendix
\section{Proof details}
\label{app:proofs}
\paragraph{Lemma~\ref{lem:parity} (even-sector confinement), expanded.}
Write $U=U_{\mathrm{ent}}\bigotimes_i \mathrm{RY}(a_i)$ with $U_{\mathrm{ent}}$ a
product of CNOTs. $\mathrm{RY}(a_i)=\exp(-i a_i Y/2)$ is a real matrix, and
$\mathrm{RY}(a_i)|0\rangle=\cos(a_i/2)|0\rangle+\sin(a_i/2)|1\rangle$ is a real
vector; CNOT is a real permutation matrix, so $|\psi(a)\rangle:=U|0\rangle$ is
real-valued for every $a$. For a $Z$-string $Z_S=\prod_{i\in S}Z_i$, conjugation by
a single CNOT with control $c$, target $t$ acts as $Z_c\!\mapsto\!Z_c$,
$Z_t\!\mapsto\!Z_cZ_t$, $Z_i\!\mapsto\!Z_i\ (i\neq c,t)$: the image is again a
$Z$-string, never producing $X$ or $Y$. Composing over all CNOTs,
$U_{\mathrm{ent}}^\dagger Z_S U_{\mathrm{ent}}=Z_{S'}$ for some set $S'$. Hence
$\langle Z_S\rangle(a)=\langle0|\bigotimes_i\mathrm{RY}(a_i)^\dagger\,Z_{S'}\,
\bigotimes_i\mathrm{RY}(a_i)|0\rangle=\prod_{i\in S'}\langle0|\mathrm{RY}(a_i)^\dagger
Z\,\mathrm{RY}(a_i)|0\rangle=\prod_{i\in S'}\cos a_i$, a product of cosines and thus
an even function of $a$. Evenness is preserved by linear (ridge) readout, so
$\mathrm{span}\,q\subseteq\{\text{even}\}$ and is $L^2$-orthogonal to every
sine-sector function. $\square$

\paragraph{Theorem~\ref{thm:floor} (containment-free parity floor), expanded.}
Fix $\sigma$ and work in $L^2(\mu_\sigma)$, $\mu_\sigma$ the marginal law of
$\theta_\sigma$. The prior $p(\theta)$ is invariant under $\theta\mapsto-\theta$
(the coupling $\cos(\theta_i+\theta_j-2\theta_k)$ is even), and the wrapped-Gaussian
kernel $\kappa_\sigma(\theta_\sigma\mid\theta_0)$ satisfies
$\kappa_\sigma(-\theta_\sigma\mid-\theta_0)=\kappa_\sigma(\theta_\sigma\mid\theta_0)$;
hence $\mu_\sigma$ is parity-symmetric and the even/odd subspaces of $L^2(\mu_\sigma)$
are orthogonal. Decompose the Bayes denoiser
$m_\sigma=\Pi_{\mathrm{even}}m_\sigma+\Pi_{\mathrm{odd}}m_\sigma$. Its
$\cos\theta_0$-coordinates, $a\mapsto\mathbb E[\cos\theta_0^i\mid\theta_\sigma=a]$,
are even in $a$ (numerator and denominator of the conditional expectation are both
even); its $\sin\theta_0$-coordinates are odd. For any $q\in Q\subseteq\{\mathrm{even}\}$,
parity orthogonality gives
\[
\|m_\sigma-q\|_{L^2(\mu_\sigma)}^2
=\|\Pi_{\mathrm{odd}}m_\sigma\|_{L^2(\mu_\sigma)}^2
+\|\Pi_{\mathrm{even}}m_\sigma-q\|_{L^2(\mu_\sigma)}^2,
\]
since $m_\sigma-q=\Pi_{\mathrm{odd}}m_\sigma+(\Pi_{\mathrm{even}}m_\sigma-q)$ splits a
vector into its odd and even parts. Taking the infimum over $q\in Q$ leaves the first
term untouched and yields
\[
R_\sigma(Q)-R^\star_\sigma
=\|\Pi_{\mathrm{odd}}m_\sigma\|_{L^2(\mu_\sigma)}^2
+\inf_{q\in Q}\|\Pi_{\mathrm{even}}m_\sigma-q\|_{L^2(\mu_\sigma)}^2
\;\ge\;\|\Pi_{\mathrm{odd}}m_\sigma\|_{L^2(\mu_\sigma)}^2 .
\]
Equality holds exactly when the infimum vanishes, i.e.\ when the even component of the
Bayes target can be approximated arbitrarily well within $Q$; if $Q$ is closed, this is
equivalent to $\Pi_{\mathrm{even}}m_\sigma\in Q$. The argument never requires $Q$ to be a
closed linear subspace, nor an orthogonal projection onto $Q$ to exist, and no containment
of $Q$ in any specific bank is used --- only $Q\subseteq\{\text{even}\}$. $\square$

\section{The CoupledPhaseTexture benchmark: generation details}
\label{app:data}
Table~\ref{tab:appdata} lists every generation parameter. The prior is the 3-body
Gibbs law $p(\theta)\propto\exp\!\big(K\sum_{(i,j,k)\in\mathcal T}
\cos(\theta_i+\theta_j-2\theta_k)\big)$ on $\mathbb T^8$, sampled by single-site
Metropolis (Gaussian proposal, acceptance $u<\exp(-\beta\,\Delta E)$ with
$E=-K\sum\cos(\cdot)$). A fresh triple set $\mathcal T$ and a fresh MCMC chain are
drawn per graph; train and test phases use independent chains under the same graph.
The forward process is the torus heat kernel realized as wrapped Gaussian
$\theta_\sigma=(\theta_0+\sigma\varepsilon)\bmod2\pi$,
$\varepsilon\sim\mathcal N(0,I_8)$.

\paragraph{Second prior (Table~\ref{tab:xy}).} The generalization benchmark uses a
2-body frustrated signed-XY prior $p(\theta)\propto\exp(\beta\sum_{(i,j)\in E}
J_{ij}\cos(\theta_i-\theta_j))$ on a random 3-regular graph ($n{=}8$),
$J_{ij}\in\{\pm1\}$ with $\approx$half negative (genuine frustration), $\beta{=}2.0$,
sampled by the same single-site Metropolis and diffused by the same wrapped
Gaussian; the decomposition, references, and readout are identical to the primary
benchmark.

\paragraph{Mixing.} We use $N$ independent chains (one per sample), each $200$
single-site Metropolis sweeps from a uniform start; acceptance is $\approx0.53$.
The single-angle first circular moments $\mathbb E[e^{i\theta_i}]$ vanish---so both
the $\cos$ \emph{and} $\sin$ marginals are $\approx0$---by the \emph{global
phase-shift invariance} of the Gibbs law (the coupling
$\cos(\theta_i+\theta_j-2\theta_k)$ is invariant under $\theta\mapsto\theta+c$), not
by inversion symmetry, which does not constrain the even $\cos$ marginal. The
quantities the benchmark actually reports are stable to chain length: as sweeps
increase $200\to400\to800$ ($\sigma{=}0.5$, 3 graphs) the parity proxy moves
$0.2105\to0.2137$ and the depth-1 quantum excess $0.2261\to0.2240$ (both $<0.004$),
and single-angle marginals drift $<0.02$. We therefore treat $200$ sweeps as
adequate for the relative comparisons; the released implementation allows
longer chains and standard diagnostics.

\begin{table}[h]\centering
\caption{CoupledPhaseTexture generation and evaluation parameters.}
\label{tab:appdata}
\begin{tabular}{ll}
\toprule
parameter & value \\
\midrule
phases $n$ (torus $\mathbb T^n$)        & $8$ \\
coupling $K$                            & $1.5$ \\
inverse temperature $\beta$             & $1.0$ \\
number of triples $|\mathcal T|$        & $12$ (distinct nodes, sampled per graph) \\
MCMC                                    & single-site Metropolis, Gaussian proposal std $0.7$ \\
MCMC sweeps                             & $200$ \\
initialization                          & $\theta\sim\mathrm{Uniform}[0,2\pi)^8$ \\
forward kernel                          & $\theta_\sigma=(\theta_0+\sigma\varepsilon)\bmod2\pi$ \\
noise grid $\sigma$                     & $0.20,\,0.35,\,0.50,\,0.70,\,1.00,\,1.40$ \\
target $Y$                              & $(\cos\theta_0,\sin\theta_0)\in\mathbb R^{16}$ \\
train / test $N$                        & $8000$ / $4000$ \\
graphs (decomposition, Table~\ref{tab:head})   & $8$ (seeds $5000$--$5007$) \\
graphs (sample complexity, Table~\ref{tab:e5}) & $6$ (seeds $2000$--$2005$) \\
graphs (design axes, Table~\ref{tab:e6})       & $6$ (seeds $3000$--$3005$) \\
graphs ($\sigma$-coupled, Table~\ref{tab:sigma})   & $6$ (seeds $4000$--$4005$) \\
graphs (second prior, Table~\ref{tab:xy})      & $6$ (seeds $7000$--$7005$) \\
graphs (reference diagnostics, Table~\ref{tab:conv}) & $3$ (seeds $9000$--$9002$; $N$-sweep $9100{+}g$) \\
ridge $\lambda$ grid                    & $\{10^{-3},10^{-2},10^{-1},1,10,100,10^3,10^4\}$ \\
$\lambda$ selection                     & $80/20$ held-out split, then refit on full train \\
\bottomrule
\end{tabular}
\end{table}

\section{Heat kernel and denoiser spectrum}
\label{app:heat}
The wrapped Gaussian is the torus heat kernel: in the Fourier basis
$\{e^{ik\theta}\}$, the forward map multiplies mode $k$ by $e^{-\sigma^2k^2/2}$.
Thus high-frequency structure of the prior is attenuated first, and the odd mass
$\|\Pi_{\mathrm{odd}}m_\sigma\|^2$ that sets the parity floor decays with $\sigma$
(Table~\ref{tab:head}, column ``parity proxy''), vanishing as $\sigma\to\infty$
where $\theta_\sigma$ carries no information about $\theta_0$. The floor's
$\sigma$-profile is therefore a property of the diffused \emph{target}, not of the
quantum map, whose encoding $U$ is $\sigma$-independent.

\section{Numerical reference and convergence}
\label{app:ref}
\label{app:conv}
The numerical reference $R_{\mathrm{ref},\sigma}$ is the degree-3 both-parity
trigonometric bank (Appendix~\ref{app:feats}) fit with the CV-$\lambda$ ridge of
Appendix~\ref{app:ridge} at $N{=}8000$. Table~\ref{tab:conv} shows (at representative
noise levels $\sigma{\in}\{0.35,0.7\}$) it is the \emph{risk-minimizing} bank over
trigonometric degrees $2$--$5$ at this sample size:
degree $2\!\to\!3$ lowers the fitted risk, but degrees $4$--$5$ yield higher
finite-sample test risk, consistent with estimation variance from the added
high-degree features (up to $3488$ at degree $5$) outweighing any reduction in
approximation bias at $N{=}8000$ (the prior carries little degree-$4^+$ structure). Thus $R_{\mathrm{ref}}$ is the best-performing tested
reference at this $N$, with $R_{\mathrm{ref},\sigma}\ge R^\star_\sigma$ (the exact
Bayes risk, and hence the gap, unknown). Increasing $N$ lets a higher degree begin to
help ($N{:}8000\!\to\!16000$ lowers the degree-4 risk from $0.110$ to $0.097$ at
$\sigma{=}0.5$), but the reference is fixed at degree-3/$N{=}8000$ and applied
identically to every compared map, so no map is advantaged by the reference choice.
The selected even reference (degree-3 cosine-only bank) is fit the same way; the
empirical parity proxy is
$\hat P_\sigma=R^{\text{even-ref}}_\sigma-R_{\mathrm{ref},\sigma}$, independent of $Q$.

\begin{table}[h]\centering
\caption{Reference convergence (full both-parity trig bank, CV-ridge, $N{=}8000$, 3
graphs). At these representative $\sigma$, degree-3 minimizes the estimated risk;
degrees $4$--$5$ yield higher finite-sample test risk at this sample size. (At large
$\sigma$, degree-2 can be lower; degree-3 is retained across the grid for protocol
consistency, not as a pointwise minimizer.)}
\label{tab:conv}
\begin{tabular}{lccc}
\toprule
degree & features & $R_{\mathrm{ref}}$ ($\sigma{=}0.35$) & $R_{\mathrm{ref}}$ ($\sigma{=}0.7$) \\
\midrule
2 & 128  & 0.057 & 0.186 \\
3 & 576  & \textbf{0.049} & \textbf{0.166} \\
4 & 1696 & 0.059 & 0.198 \\
5 & 3488 & 0.083 & 0.272 \\
\bottomrule
\end{tabular}
\end{table}

\section{Ridge protocol}
\label{app:ridge}
All banks use the same readout: standardize features by train mean/std, then
closed-form ridge $W=(X^\top X+\lambda I)^{-1}X^\top Y$. The penalty $\lambda$ is
selected on a held-out $80/20$ split of the training set over
$\{10^{-3},\dots,10^4\}$, then $W$ is refit on the full training set at the selected
$\lambda$. Held-out $\lambda$ selection reduces regularization-selection confounds and
provides a consistent best-in-bank comparison across feature families (removing the
over/under-fit artifact that a single fixed $\lambda$ introduces). Reported risks are
test MSE on $Y\in\mathbb R^{16}$.

\section{Feature banks}
\label{app:feats}
Table~\ref{tab:appfeats} gives exact feature counts. Trigonometric banks are built
over phase combinations: degree-1 ($a_i$, $8$ phases), degree-2 ($a_i\pm a_j$,
$2\binom{8}{2}=56$ phases), degree-3 ($\pm a_i\pm a_j\pm a_k$ up to overall sign,
$4\binom{8}{3}=224$ phases); the even bank takes $\cos$ of each phase, the full
bank takes $\cos$ and $\sin$. Quantum maps use exact statevector expectations
(8 qubits, $256$-dim, \texttt{shots}=None) of an $\mathrm{RY}$ angle encoding with
a CNOT ring, data re-uploaded once (depth-1) unless the depth axis is swept.

\begin{table}[h]\centering
\caption{Feature-bank dimensions.}
\label{tab:appfeats}
\begin{tabular}{lll}
\toprule
bank & readout / phases & dim \\
\midrule
quantum even ($q_{\text{even}}$)   & $\expct{Z_i}$, $\expct{Z_iZ_j}$ & $36$ \\
quantum $X{+}Z$ ($q_{\text{full}}$) & $\expct{X_i},\expct{Z_i},\expct{Z_iZ_j},\expct{X_iZ_j},\expct{X_iX_j}$ & $100$ \\
classical even                     & $\cos$, degree-2 & $64$ \\
classical full                     & $\cos,\sin$, degree-2 & $128$ \\
selected even reference            & $\cos$, degree-3 & $288$ \\
numerical reference $R_{\mathrm{ref}}$ & $\cos,\sin$, degree-3 & $576$ \\
\bottomrule
\end{tabular}
\end{table}

\section{Design-axis and re-upload conditioning}
\label{app:design}
The design-axis ablation (Table~\ref{tab:e6}) varies readout order, entangler, and
re-upload depth independently ($\sigma{=}0.5$, $6$ graphs). Re-uploading does
\emph{not} change the readout dimension---the $\expct{Z_i}/\expct{Z_iZ_j}$ feature
count stays at $36$ for every depth---so the non-monotone depth trend is not a
capacity-count effect. Table~\ref{tab:reup} instruments it directly ($N{=}8000$, 3
graphs, ridge $\lambda$ grid extended to $10^4$): from depth $1$ to $3$ the
design-matrix condition number \emph{decreases} ($9.0\to4.7\to2.9$) and the selected
$\lambda$ \emph{rises} (to $1000$ at depth $3$, versus $100$ at depths $1$--$2$),
while \emph{both} train and test excess jump together (train $0.24\to0.23\to0.41$,
test $0.23\to0.22\to0.40$; the small train$-$test gap stays $\approx0.01$ at every
depth and does not widen at depth $3$). The joint increase of train and test excess,
with no widening train/test gap, despite improved conditioning and heavier
regularization, is consistent with representation misalignment rather than
ill-conditioning or overfitting: re-uploading pushes the
$36$ readouts to high Fourier frequency, poorly aligned with the smooth denoising
target, so a linear ridge head underfits even in-sample and $\lambda$ shrinks it
toward the mean. Depth trades useful low-frequency content for high-frequency
content the readout cannot use, which is a representation-alignment effect rather
than a capacity gain.

\begin{table}[h]\centering
\caption{Re-upload conditioning ($\sigma{=}0.5$, $N{=}8000$, 3 graphs; readout fixed
at $36$ dims). Deeper re-uploading \emph{improves} conditioning yet raises train and
test excess together, arguing against ill-conditioning and conventional train-test
overfitting as the primary explanation.}
\label{tab:reup}
\begin{tabular}{lccccc}
\toprule
depth & readout dim & cond.\ number & selected $\lambda$ & train excess & test excess \\
\midrule
1 & 36 & 9.0 & 100 & 0.241 & 0.230 \\
2 & 36 & 4.7 & 100 & 0.233 & 0.223 \\
3 & 36 & 2.9 & 1000 & 0.410 & 0.401 \\
\bottomrule
\end{tabular}
\end{table}

\section{Reproducibility}
\label{app:repro}
Every result uses fixed NumPy default-RNG streams, seeded per experiment as in
Table~\ref{tab:appdata}: headline decomposition $5000{+}g$, sample-complexity
sweep $2000{+}g$, design-axis ablation $3000{+}g$, $\sigma$-coupled comparison
$4000{+}g$, second-prior benchmark $7000{+}g$, MCMC-mixing diagnostics $6000{+}g$,
reference diagnostics $9000{+}g$ (reference $N$-sweep $9100{+}g$), together with the
configuration described in this appendix. All
quantum expectations are exact statevector values (no shot noise). The
implementation---benchmark generator, feature banks, ridge readout, and
decomposition instrument---will be released publicly so that alternative fixed
encoder/readout/entangler choices can be evaluated under the same protocol.

\end{document}